\documentclass[11pt]{article}

\usepackage{amsmath,amssymb,bbm,amsthm}
\usepackage{fullpage}
\usepackage{thm-restate,color,xspace}
\usepackage{comment}
\usepackage{thmtools}

\usepackage{amsmath,amsfonts,dsfont,amsthm,amssymb}
\usepackage{color}
\usepackage[linesnumbered, lined, ruled]{algorithm2e}
\usepackage{enumerate}
\usepackage{graphicx}

\usepackage{xcolor}
\usepackage{nameref}
\definecolor{ForestGreen}{rgb}{0.1333,0.5451,0.1333}
\definecolor{DarkRed}{rgb}{0.65,0,0}
\definecolor{Red}{rgb}{1,0,0}
\usepackage[linktocpage=true,
pagebackref=true,colorlinks,
linkcolor=DarkRed,citecolor=ForestGreen,
bookmarks,bookmarksopen,bookmarksnumbered]
{hyperref}
\usepackage{cleveref}

\usepackage{xspace}

\declaretheorem[numberwithin=section]{theorem}
\declaretheorem[numberlike=theorem]{lemma}

\declaretheorem[numberlike=theorem]{fact}
\declaretheorem[numberlike=theorem]{proposition}

\declaretheorem[numberlike=theorem]{corollary}

\declaretheorem[numberlike=theorem,style=definition]{definition}
\declaretheorem[numberlike=theorem,name=Definition,style=definition]{defn}

\newcommand{\e}{\epsilon}

\newcommand{\tO}{\tilde{O}}
\newcommand{\polylog}{{\rm poly}\log}
\newcommand{\poly}{{\rm poly}}

\def\polylog{\text{polylog}}

\def\BAL#1\EAL{\begin{align*}#1\end{align*}}
\def\BALN#1\EALN{\begin{align}#1\end{align}}
\def\BG#1\EG{\begin{gather}#1\end{gather}}

\newcommand{\eat}[1]{}

\graphicspath{ {./images/} }

\newcommand{\dc}{DECA\xspace}

\begin{document}

\title{Edge Connectivity Augmentation in Near-Linear Time}
\author{Ruoxu Cen\thanks{Department of Computer Science, Duke University.}\and Jason Li\thanks{Simons Institute for Theory of Computing, University of California, Berkeley.}\and Debmalya Panigrahi\thanks{Department of Computer Science, Duke University.}}

\pagenumbering{gobble}

\maketitle

\begin{abstract}
    We give an $\tO(m)$-time algorithm for the edge connectivity augmentation problem and the closely related edge splitting-off problem. This is optimal up to lower order terms and closes the long line of work on these problems. 
\end{abstract}

\maketitle

\pagenumbering{arabic}

\section{Introduction}\label{sec:intro}


In the {\em edge connectivity augmentation} problem, we are given an undirected graph $G = (V, E)$ with edge weights $w$, and a target connectivity $\tau > 0$. The edge weights $w$ and connectivity target $\tau$ are assumed to be polynomially bounded integers. The goal is to find a minimum weight set $F$ of edges on $V$ such that adding these edges to $G$ makes the graph $\tau$-connected. (In other words, the value of the minimum cut of the graph $G' = (V, E\cup F)$ should be at least $\tau$.) The edge connectivity augmentation problem is known to be tractable in $\poly(m, n)$ time, where $m$ and $n$ denote the number of edges and vertices respectively in $G$. This was first shown by Watanabe and Nakamura~\cite{WatanabeN87} for unweighted graphs, and the first strongly polynomial algorithm was obtained by Frank~\cite{Frank92}. Since then, several algorithms~\cite{CaiS89,NaorGM97,Gabow16,Gabow94,NagamochiI97} progressively improved the running time and till recently, the best known result was an $\tO(n^2)$-time\footnote{$\tO(\cdot)$ ignores (poly)-logarithmic factors in the running time.} algorithm due to Bencz\'ur and Karger~\cite{BenczurK00}. This was improved to the current best runtime of $\tO(m+n^{3/2})$~\cite{CenLP22}
by reducing the edge connectivity augmentation problem to $\polylog(n)$ max-flow calls. The runtime bound follows from the current best max-flow algorithm on undirected graphs~\cite{BrandLLSSSW21}.\footnote{We note that for sparse graphs, there is a slightly faster max-flow algorithm that runs in $O(m^{3/2-\delta})$ time~\cite{GaoLP21}, where $\delta > 0$ is a small constant. If this max-flow algorithm is used in~\cite{CenLP22}, a running time of $O(m^{3/2-\delta})$ is obtained for the augmentation problem.} This represents a natural bottleneck for the problem since further improvement would need techniques that do not rely on max-flows. 

We overcome this bottleneck in this paper, and obtain a nearly-linear $\tO(m)$-time algorithm for the edge connectivity augmentation problem. This is optimal up to poly-logarithmic terms, and brings to an end the long line of work on this problem (barring further improvements in the logarithmic terms). Moreover, it demonstrates that this problem is {\em easier} than max-flow, since obtaining an $\tO(m)$-time max-flow algorithm remains a major open problem. We state our main result below:


\begin{theorem}
    \label{thm:augment}
    There is an $\tO(m)$-time randomized Monte Carlo algorithm for the edge connectivity augmentation problem.
\end{theorem}

The edge connectivity augmentation problem is closely related to {\em edge splitting off}, a widely used tool in the graph connectivity literature (e.g.,~\cite{Gabow94,NagamochiI97}). A pair of (weighted) edges $(u, s)$ and $(s, v)$ both incident on a common vertex $s$ is said to be split off by weight $w$ if we reduce the weight of both these edges by $w$ and increase the weight of their {\em shortcut} edge $(u, v)$ by $w$. Such a splitting off is valid if it does not change the (Steiner) connectivity\footnote{The Steiner connectivity of a set of vertices is the minimum value of any cut that has at least one of these vertices on each side of the cut.} of the vertices $V\setminus \{s\}$. 
If all edges incident on $s$ are eliminated by a sequence of splitting off operations, we say that the vertex $s$ is split off. We call the problem of finding a set of shortcut edges to split off a given vertex $s$ the edge splitting off problem. 

Lov\'asz~\cite{Lovasz79} initiated the study of edge splitting off by showing that in an undirected graph, any vertex $s$ with even degree (i.e.\ the total weight of incident edges is even) can be split off while maintaining the (Steiner) connectivity of the remaining vertices. (Later, more powerful splitting off theorems~\cite{Mader78} were obtained that preserve stronger properties and/or apply to directed graphs, but these come at the cost of slower algorithms. We do not consider these extensions in this paper.) The splitting off operation has emerged as an important inductive tool in the graph connectivity literature, leading to many algorithms with progressively faster running times being proposed for the edge splitting off problem~\cite{CaiS89,Frank92,Gabow94,NagamochiI97,BenczurK00}. Currently, the best running time is $\tO(m + n^{3/2})$, which was obtained in the same paper as the current best edge connectivity augmentation result~\cite{CenLP22}. We improve this bound as well:
\begin{theorem}\label{thm:splitting}
    There is a randomized, Monte Carlo algorithm for the edge splitting off problem that runs in $\tO(m)$ time.
\end{theorem}

A key tool in augmentation/splitting off algorithms (e.g., in \cite{WatanabeN87,NaorGM97,Gabow16,Benczur94,BenczurK00,CenLP22}) is that of {\em extreme sets}. A non-empty set of vertices $X\subsetneq V$ is called an extreme set in graph $G = (V, E)$ if for every non-empty proper subset $Y\subsetneq X$, we have $\delta_G(Y) > \delta_G(X)$, where $\delta_G(X)$ (resp., $\delta_G(Y)$) is the total weight of edges with exactly one endpoint in $X$ (resp., $Y$) in $G$. (If the graph is unambiguous, we drop the subscript $G$ and write $\delta(\cdot)$.) In edge connectivity augmentation problem, every vertices set $U$ with $\delta(U)<\tau$ has a demand that the solution must add edges with total weight at least $\tau-\delta_G(U)$ across $U$. It turns out that satisfying the demands of all extreme sets implies satisfying the demands of all vertices sets. 
The extreme sets form a laminar family, thereby allowing an $O(n)$-sized representation in the form of an {\em extreme sets tree}. The main bottleneck of the previous edge augmentation/splitting off algorithms~\cite{BenczurK00,CenLP22} is in the construction of the extreme sets tree. Indeed, \cite{CenLP22} show that once the extreme sets tree is constructed, the augmentation/splitting off problems can be solved in $\tO(m)$ time:
\begin{theorem}[Theorem 3.1 in \cite{CenLP22}]\label{thm:reduction}
    Given an algorithm to compute the extreme sets tree, the edge connectivity augmentation and edge splitting problems can be solved in $\tO(m)$ time.
\end{theorem}
 \noindent (Bencz\'ur and Karger~\cite{BenczurK00} also hint that computing extreme sets is the main bottleneck of their algorithm, although their algorithm does use $\tO(n^2)$ time in a few other places.)

\Cref{thm:reduction} reduces \Cref{thm:augment} and \Cref{thm:splitting} to obtaining an extreme sets tree in $\tO(m)$ time. 
Bencz\'ur and Karger~\cite{BenczurK00} used the {\em recursive contraction} framework of Karger and Stein~\cite{KargerS96} to construct the extreme sets tree, which takes $\tO(n^2)$ time. This was improved by Cen, Li, and Panigrahi~\cite{CenLP22} who used the isolating cuts framework~\cite{LiP20deterministic}\footnote{A similar framework was shown independently by Abboud, Krauthgamer, and Trabelsi~\cite{AbboudKT21}.} which uses $\polylog(n)$ max-flow calls. But the isoating cuts framework is unusable if we want to improve beyond max-flow runtime, since a special case of an isolating cut is an $s-t$ min-cut. In this paper, we overcome this barrier and give an $\tO(m)$-time algorithm for finding the extreme sets tree of a graph:
\begin{restatable}{theorem}{Extreme}\label{thm:extreme}
    There is a randomized, Monte Carlo algorithm for finding the extreme sets tree of an undirected graph that runs in $\tO(m)$ time.
\end{restatable}


Given \Cref{thm:reduction}, the rest of this paper focuses on proving \Cref{thm:extreme}. 

\subsection{Our Techniques}
Our $\tO(m)$-time extreme sets algorithm can be viewed as a series of reductions to finding extreme sets in progressively simpler settings. Recall that the original problem is to find extreme sets in an arbitrary undirected graph. Our first step is an {\em iterative refinement} of this problem, namely instead of finding all extreme sets, we refine the problem to finding extreme sets that are also nearly minimum cuts (we call these {\em near-mincut extreme sets}). More precisely, suppose we have identified all extreme sets whose cut value is at most some threshold $\gamma$. These extreme sets form a laminar family, and induce an equivalence partition on the vertices where any two vertices that are not separated by any of these extreme sets are in the same set of the partition. By laminarity and the extreme sets property, we can claim that all the extreme sets whose cut value exceeds $\gamma$ must be strict subsets of the sets in this equivalance partition. This justifies a natural recursive strategy: for each set in the equivalence partition, we contract all the vertices outside this set into a single vertex and find extreme sets in this contracted graph. 

So far, we have reduced the problem to finding extreme sets that are subsets of the set of (uncontracted) vertices $S$, where in addition, there is a (contracted) vertex $c$ representing all the other vertices $V\setminus S$ of the graph. Clearly, the Steiner connectivity of $S$, denoted $\lambda$, must exceed $\gamma$ (else, any minimum Steiner cut that is also minimal in terms of vertices would also be an extreme set of cut value $\gamma$, which contradicts the inductive assumption that we have already identified all extreme sets of cut value at most $\gamma$). We now define our iterative goal: find all extreme sets that are subsets of $S$ and have cut value in the range $[\lambda, (1+\epsilon)\lambda]$.

To solve this problem, our next pair of tools is sparsification and tree packing. Suppose, for now, that the minimum cut in the graph containing $S$ and the contracted vertex $c$ is of value $\lambda$ (this may not be true in general since the degree cut for vertex $c$ can be smaller than $\lambda$, but we will handle this complication later). Then, we can use a {\em uniform sparsification} technique of Karger~\cite{Karger99} to sample edges and form a graph where the value of all cuts converge to their expected value {\em whp}\footnote{with high probability, i.e., with probability $1-o(1)$} and where the expected value of the minimum cut is $O(\log n)$. On this graph, we can pack $O(\log n)$ disjoint spanning trees rooted at a fixed vertex $r$ in $\tO(m)$ time such that the following property holds whp: for every cut of value at most $(1+\epsilon)\lambda$ in the original graph (this includes all the extreme sets we are interested in finding), there is a spanning tree that contains at most two edges from the cut (we say the cut {\em $2$-respects} the tree). This essentially reduces the problem to finding extreme sets that $2$-respect a given spanning tree. There are two caveats. First, we need an algorithm that can {\em merge} the extreme sets found for the different trees into a single extreme sets tree. Second, the contracted vertex $c$ may have degree smaller than the number of trees, which means that trees wouldn't be spanning and two tree edges do not uniquely define a vertex set. We defer the technical details to address these issues to later sections.

Now, we have reduced the extreme sets problem to finding all extreme sets that $2$-respect a given tree. If we were interested in finding a minimum cut (see~\cite{Karger00}), then we would use a dynamic program at this stage. But, extreme sets are more complex. First, the extreme set condition is more difficult to check than tracking the minimum cut. More importantly, extreme sets are asymmetric, i.e., even if $X$ is an extreme set, $V\setminus X$ may or may not be an extreme set. This seems to defeat the purpose of working on a tree. For instance, when the two tree edges are comparable, i.e., form an ancestor-descendant pair, the extreme set is not contiguous in the tree. It is unclear at all how we can check the extreme sets property for such a non-contiguous set. To overcome these difficulties, we undertake two further simplifications of the problem. The first is a {\em recursive rotation} of the tree based on the idea of a centroid decomposition. We show that this ensures that all $2$-respecting cuts will appear as two incomparable edges in some level of the recursion, thereby eliminating the need for handling comparable tree edges. Our next technique reduces the problem from trees with arbitrary structure to spiders. (A {\em spider} is a tree where only the root can have degree greater than $2$.) The basic idea is to perform a {\em heavy-light decomposition} of the spanning tree, and then sample each path of this decomposition independently for contraction in a manner that the resulting tree is a spider. If this process is repeated $O(\log n)$ times, then for every $2$-respecting cut, there whp is at least one spider that preserves both edges of the cut in the tree. (This idea was previously explored by Li~\cite{Li19}, although in a somewhat different context.) As an aside, we note that both these simplifications are also valid for the minimum cut problem, and can be used to simplify Karger's celebrated near-linear time minimum cut algorithm~\cite{Karger00}.

We have now reduced the problem to finding $2$-respecting extreme sets on spiders, with the additional guarantee that if the cut contains exactly two edges in the tree, then those two edges will be incomparable. At this point, we first find all $1$-respecting extreme sets using a simple dynamic program. Conceptually, this is simple because we can run the algorithm ``in parallel'' on each branch of the spider. However, the $2$-respecting case still needs additional work. At this point, we use our final simplification, where we reduce the $2$-respecting extreme sets problems from spiders to paths (equivalently, spiders with only two branches). The basic idea behind this transformation is that we use the laminar structure of extreme sets to claim that all $2$-respecting extreme sets can be partitioned into equivalence classes, where each set of the partition corresponds to two distinct branches of the spider. This allows us to run the $2$-respecting algorithm ``in parallel'' on these spiders containing only two branches each, i.e., on paths. Finally, for each path, we can solve the $2$-respecting extreme sets problem using a simple dynamic program.

\paragraph{Roadmap.} We introduce some preliminaries in \Cref{sec:prelim}. \Cref{sec:iterative} describes the iterative framework that we use in our extreme sets algorithm, and reduces the problem to finding $2$-respecting extreme sets for a spanning tree of the graph. In \Cref{sec:tree2spider}, we use the recursive rotation based on centroid partitioning and the random sampling over the heavy-light decomposition to reduce the problem to finding $2$-respecting extreme sets in a spider. We solve this latter problem in \Cref{sec:respect}, using the reduction to a path and the employing a dynamic program. Finally, in \Cref{sec:merge}, we give the algorithm to merge the extreme sets revealed by the different steps into a single extreme sets tree. 

\section{Preliminaries}
\label{sec:prelim}
Use $\delta(S)$ to denote the value of a cut $S\subsetneq V$, that is the sum of weights of edges with exactly one endpoint in $S$. For disjoint $S, T\subsetneq V$, denote $\delta(S, T)$ to be the sum of weights of edges with one endpoint in $S$ and the other endpoint in $T$. For vertices $s\ne t$, denote $\lambda(s,t)$ to be the value of minimum $s$-$t$ cut.

Our goal is to find all the extreme sets of an undirected graph $G = (V, E)$. We can define an extreme set as follows.
\begin{defn}[Extreme set]
A nonempty set $X\subsetneq V$ is {\em extreme} if for every non-empty proper subset $U$ of $X$, we have $\delta(U)>\delta(X)$. By convention, all singleton sets are extreme sets.
\end{defn}
One noteworthy aspect of this definition is that although the graph $G$ is undirected, the notion of extreme sets is {\em asymmetric}. In other words, if $X$ is an extreme set, it is not necessarily the case that the complementary set $V\setminus X$ is also an extreme set. As described in the introduction, this asymmetry is one of the main contributors to the difficulty of the problem.

We need the following properties of cut function in undirected graphs.
\begin{proposition}[submodularity]
$\forall X, Y\subseteq V, \delta(X\cap Y)+\delta(X\cup Y)\le \delta(X)+\delta(Y)$.
\end{proposition}
\begin{proposition}[posi-modularity]
$\forall X, Y\subseteq V, \delta(X\setminus Y)+\delta(Y\setminus X)\le \delta(X)+\delta(Y)$.
\end{proposition}

A family of sets is said to be {\em laminar} if any two of them are either disjoint or one is contained in the other. It is well known that extreme sets form a laminar family.

\begin{lemma}\label{lem:laminar}
Extreme sets form a laminar family.
\end{lemma}
\begin{proof}
Assume for contradiction that there are two extreme sets $X$ and $Y$ violate laminarity, i.e., $X\setminus Y, Y\setminus X,$ and $X\cap Y$ are all non-empty sets.
Then, since both $X$ and $Y$ are extreme sets, we have $\delta(X\setminus Y) > \delta(X)$ and $\delta(Y\setminus X) > \delta(Y)$.
Then $\delta(X\setminus Y) + \delta(Y\setminus X) > \delta(X) + \delta(Y)$, which contradicts posi-modularity of the cut function.
\end{proof}

Laminarity induces a natural tree structure on extreme sets where all the vertices of the graph (as trivial extreme sets) are leaves of the tree and every subtree (or equivalently, the internal tree node where the subtree is rooted) represents an extreme set containing all the vertices that are leaves in the subtree. We call this the {\em extreme sets tree}. Our goal in this paper is to find an extreme sets tree in $\tO(m)$ time, thereby establishing \Cref{thm:extreme}.

We also use the notion of {\em Steiner connectivity} of a set of vertices, which is the minimum value of a cut that has at least one terminal on each side of the cut. If we remove this additional condition (equivalently, set all vertices as terminals), then we get the {\em edge connectivity} of the graph.
\begin{defn}[Steiner connectivity]
The Steiner connectivity of a set of vertices $X\subseteq V$ (called terminals) is the minimum value of a cut $S$ such that $X\cap S$ and $X\setminus S$ are both nonempty. If $X = V$, then we call this the edge connectivity of the graph.
\end{defn}

\section{Reduction to 2-respecting Extreme Sets}
\label{sec:iterative}

In this section, we reduce the problem of finding all extreme sets to that of finding extreme sets that satisfy an additional property called $2$-respecting that we will define later. This reduction is in two parts. In the first part, we use a framework that iteratively calls an algorithm to find all extreme sets whose cut values are in a given range. In the second part, we reduce from the problem of finding all extreme sets in a given range of cut values to all extreme sets that satisfy the $2$-respecting property.

\subsection{Iterative Framework for Extreme Sets Algorithm}

We use an iterative framework to find all extreme sets of the graph. In fact, consider the following reformulation of this problem. Given a set of vertices $S\subsetneq V$, we need to find all extreme sets that are subsets of $S$ (including $S$ itself if it is an extreme set). We note that this problem is actually equivalent to the problem of finding all extreme sets in the graph. In one direction, an algorithm that finds all extreme sets also identifies those that are subsets of $S$. But, also conversely, we can add a dummy isolated vertex to the graph, and then set $S = V$ to find all extreme sets of the graph. 

We further refine the task of finding extreme sets contained in $S$ into finding extreme sets whose cut value is in the range $[\lambda, (1+\e)\lambda]$ for a fixed constant $\e > 0$. Here, $\lambda$ is the Steiner connectivity of $S$ after we contract $V\setminus S$ into a single vertex $c$. We call these {\em near-mincut extreme sets}.

\begin{defn}[Near-mincut Extreme Set]
Suppose we are given an undirected graph and a set of vertices $S$. Let $\lambda$ denote the Steiner connectivity of $S$ when $V\setminus S$ is contracted to a single vertex $c$. Given a fixed constant $\e > 0$ (whose precise value will be given in \Cref{thm:sparsify-small-deg}), a near-mincut extreme set $S'$ is an extreme set that is a subset of $S$ (i.e., $S'\subsetneq S$) and whose cut value satisfies $\delta(S') \in [\lambda, (1+\e)\lambda)$.
\end{defn}

In the rest of this section, we describe an algorithm to find all extreme sets contained in $S$ by iteratively using an algorithm that finds near-mincut extreme sets. To describe our algorithm, it is convenient to partition cuts based on a threshold $d$ into {\em $d$-strong} and {\em $d$-weak} cuts. 
\begin{defn}[$d$-Strong and $d$-Weak Cuts]
    A nonempty set of vertices $X\subsetneq V$ is said to be $d$-strong if the cut value $\delta(X) \ge d$, else it is said to be $d$-weak.
\end{defn}

Note that the problem of finding all near-mincut extreme sets is equivalent to that of finding all $(1+\epsilon)\lambda$-weak extreme sets after contracting $V\setminus S$ into a single vertex $c$. In \Cref{alg:iterative}, we use a subroutine that returns all $(1+\e)\lambda$-weak extreme sets to obtain all extreme sets contained in $S$. Since the near-mincut extreme sets form a laminar family (by \Cref{lem:laminar}), these $(1+\epsilon)\lambda$-weak extreme sets induce a {\em canonical partition} of the vertices of $S$ defined below.

\begin{defn}[Canonical Partition]
Define an equivalence relation on the vertices of $S$ using the following rule: two vertices are related if and only if they are not separated by any of the $(1+\e)\lambda$-weak extreme sets contained in $S$. The equivalence classes corresponding to this equivalence relation form the canonical partition of $S$.
\end{defn}

The following lemma asserts that all $(1+\e)\lambda$-strong extreme sets contained in $S$ must {\em respect} this canonical partition.
\begin{lemma}\label{lem:canonical}
    Any $(1+\e)\lambda$-strong extreme sets contained in $S$ must be contained in some equivalence class of the canonical partition. 
\end{lemma}
\begin{proof}
Suppose not, and let $u, v\in S$ be two vertices in different equivalence classes of the canonical partition that are both in some $(1+\e)\lambda$-strong extreme set $S'$. By definition of the equivalence relation, there must be some $(1+\e)\lambda$-weak extreme set $X\subsetneq S$ such that $u\in X, v\notin X$ (or vice-versa). By \Cref{lem:laminar}, it must be that $X\subsetneq S'$ since $u\in X\cap S'$. But, this violates the fact that $S'$ is an extreme set since $\delta(X) < (1+\e)\lambda \le \delta(S')$.
\end{proof}

This lemma allows us to recurse on the individual equivalence classes of the canonical partition in \Cref{alg:iterative}.

\begin{algorithm}
\caption{Iterative Framework for Extreme Sets}
\label{alg:iterative}
\SetKwInOut{Input}{Input}
\SetKwInOut{Output}{Output}
\Input{Graph $G=(V,E)$ and a set $S\subsetneq V$.}
\Output{The family of all extreme sets in $G$ that are contained in $S$.}
When $S$ is a singleton, return the singleton as the only extreme set contained in $S$.\\
Let $\lambda$ denote the Steiner connectivity of $S$ after contracting $V\setminus S$ into a single vertex $c$. Let $d=(1+\epsilon)\lambda$.\\
 Call the near-mincut extreme sets subroutine to find all $d$-weak extreme sets contained in $S$. This induces a canonical partiton of $S$ into subsets $S_1, S_2, \ldots, S_k$.\label{weak}\\
For each set $S_i$, recursively find all extreme sets contained in $S_i$.\\
Construct a laminar family of all extreme sets in the current call and all extreme sets found in the recursive calls. Return the laminar family.\label{simple-merge} 
\end{algorithm}

\begin{figure}[t]
\centering
\includegraphics[width=.9\textwidth]{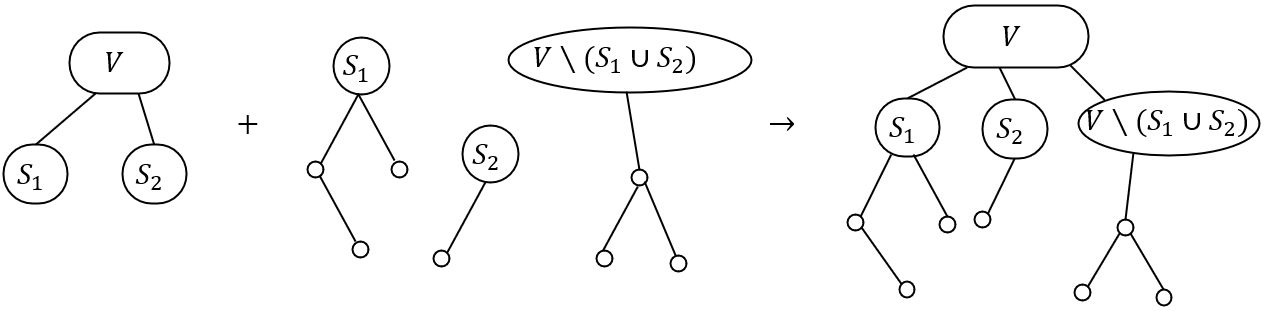}
\caption{Implementation of line~\ref{simple-merge} of \Cref{alg:iterative}. Left: $d$-weak extreme sets. Middle: extreme sets trees returned by recursive calls. Right: the merged extreme sets tree.}
\label{fig:simple-merge}
\end{figure}

\begin{theorem}
\Cref{alg:iterative} finds all extreme sets that are contained in $S$. 
\end{theorem}
\begin{proof}
The proof is by induction on the size of $S$. When $|S|=1$, the singleton set is the only extreme set.
Next consider $|S|\ge 2$. Any $d$-weak extreme set contained in $S$ will be found by the near-mincut extreme sets subroutine. Consider any $d$-strong extreme set $S'$. By \Cref{lem:canonical}, such an extreme set must be contained in one of the equivalence classes of the canonical partition. To apply the inductive hypothesis asserting that $S'$ will be revealed in a recursive call made by the algorithm, we need to show that the equivalence classes of the canonical partition are proper subsets of $S$, i.e., they are strictly smaller than $S$. This is because there is at least one cut of value $\lambda$ that is contained in $S$, since $\lambda$ is the Steiner connectivity of $S$ after contracting $V\setminus S$ into a single vertex $c$. Now, if we consider any minimal subset of $S$ of cut value $\lambda$, it must be an extreme set by definition. Therefore, the canonical partition is nontrivial, i.e., it contains at least two equivalence classes. Consequently, each set in the equivalence partition is a strict subset of $S$.

We also need to verify that any recursive call on a set $S_1\subsetneq S$ does not return spurious extreme sets, i.e., sets that are extreme in the graph where $V\setminus S_1$ is contracted, but are not extreme in the original graph. But, this can be ruled out based on the definition of extreme sets since the property only depends on the cut values of subsets of $S_1$ which are unaffected by the contraction.
\end{proof}

\eat{
\begin{theorem}
\Cref{alg:iterative} only outputs extreme sets.
\end{theorem}
\begin{proof}
Let $U$ be any set in the output. $U$ must be found by a recursive call on some $S$. That means $U\subseteq S$ is a extreme set in the graph where $V\setminus S$ is contracted. Notice that $S$ is not contracted, and the definition of extreme sets is local, so the conditions are the same after contracting $V\setminus S$, and $U$ is also a extreme set in the original graph.
\end{proof}
}

We now bound the running time for the overall algorithm.

\begin{theorem}
If we can find all near-mincut extreme sets in $\tO(m)$ time, then \Cref{alg:iterative} finds all extreme sets contained in $S$ in $\tO(m)$ time.
\end{theorem}
\begin{proof}
In each recursive level, the uncontracted vertices form a disjoint partition of $S$. Thus, each edge of the graph appears in at most 2 subproblems. So each recursive level has $O(m)$ edges across all subproblems, and therefore, takes $\tO(m)$ time by induction. 

To bound the depth of the recursion, we compare the value of $\lambda$ between a subproblem with set $S$ (call this $\lambda(S)$) and its child subproblem with set $S_i$ (call this $\lambda(S_i)$). We claim: $\lambda(S_i)\ge (1+\e)\lambda(S)$. Suppose not; then, there is a proper subset of $S_i$ that has cut value $< (1+\e)\lambda$. Now, any minimal subset (call it $S'_i$) of $S_i$ with cut value $< (1+\e)\lambda$ must be an extreme set by definition. But, now if we choose two vertices $u, v\in S_i$ where $u\in S'_i, v\notin S'_i$, then $u$ and $v$ cannot be in the same equivalence class of the canonical partition since that would contradict the fact that $S'_i$ is a $(1+\e)\lambda$-weak extreme set contained in $S$. This implies that $\lambda(S_i)\ge (1+\e)\lambda(S)$. This bounds the depth of recursion in \Cref{alg:iterative} to $O(\e^{-1}\log n)$ since the edge weights are polynomially bounded. 

Finally, we need to give an implementation of line~\ref{simple-merge} of \Cref{alg:iterative} (see \Cref{fig:simple-merge}). We map each set of the canonical partition to a unique node in the $d$-weak extreme sets tree (call it $T$) returned by line~\ref{weak}. This can be done naturally by mapping every vertex in $S$ to the smallest extreme set that it belongs to among the $d$-weak ones. (All vertices in $S$ that are not in any $d$-weak extreme set are mapped to the root representing $V$.) Note that by definition of the canonical partition, the recursive calls are on sets of graph vertices that are mapped to the same node in $T$. Consider a recursive call for a set $X$. $X$ is mapped to a node $x$ representing $X'\supseteq X$ in $T$. The recursive call returns an extreme sets tree $T'$ whose root represents $X$. If $X\subsetneq X'$, we attach $T'$ as a child of $x$ in $T$; If $X=X'$, we attach the children of the root of $T'$ as children of $x$ in $T$. 
Note that this can done in $O(n)$ time across all the recursive calls because the corresponding extreme set trees are disjoint.

The total time complexity of \Cref{alg:iterative} is then given by $\tO(m)$.
\end{proof}

\subsection{Sparsification and Tree Packing}
We further reduce near-mincut extreme sets to 2-respecting extreme sets via tree packing. We start with the following uniform sampling theorem.

\begin{theorem}[\cite{Karger99}]
\label{lem:karger-sample}
Given a weighted undirected graph $G$ with min-cut value $\lambda$ and any constant $\epsilon \in (0, 1)$, we can construct in $O(m)$ time a subgraph $H$ such that the following holds whp: for every cut $S$ in $H$, its value in $H$ (denoted $\delta_H(S)$) and its value in $G$ (denoted $\delta(S)$) are related by $\delta_H(S)\in[(1-\epsilon) p\cdot \delta(S), (1+\epsilon) p\cdot \delta(S)]$, where $p=O\left(\frac{\log n}{\lambda}\right)$. Note that this implies that the min-cut value in $H$ is $O(\log n)$.
\end{theorem}

\eat{
\begin{lemma}[\cite{NagamochiI92}]
\label{lem:NI-forest}
Given an unweighted multigraph $G$, in $O(k|E|)$ time we can construct a subgraph $H$ with $O(nk)$ edges, such that for every vertices pair $(u,v)$, $\lambda_H(u,v)\ge\min\{\lambda_G(u,v), k\}$.
\end{lemma}
}

First, we use this theorem to prove the following lemma on sampling graphs to preserve near-mincut extreme sets.
\begin{lemma}
\label{thm:sparsify-small-deg}
Given a weighted undirected graph $G = (V = S\cup \{c\}, E)$, we can construct in $O(m)$ time a subgraph $H$ where the following hold whp: (a) the Steiner min-cut value of $S$ in graph $H$ is $\lambda_H = O(\log n)$, and (b) every near-mincut extreme set in $G$ has cut value at most $1.1\lambda_H$ in $H$.
\end{lemma}

\begin{proof}
Let $\lambda$ be the Steiner min-cut value of vertices $S$ in graph $G$. Choose $\epsilon=0.01$. Let $\delta(c)$ denote the value of the singleton cut $\{c\}$ in graph $G$.

When $\delta(c)\ge \epsilon \lambda$, we use \Cref{lem:karger-sample} to get a graph $H_1$ with min cut value $\lambda_1=O(\log n)$. We have 
$$\lambda_1\ge (1-\epsilon) p\cdot \min\{\lambda,\delta(c)\}\ge \epsilon(1-\epsilon)p\lambda,$$ 
which implies that $p\lambda=O(\log n)$. 
The Steiner min-cut value of $S$ in $H$ is 
$$\lambda_H\in [(1-\epsilon)p\lambda,(1+\epsilon)p\lambda]=O(\log n).$$ 
For any near-mincut extreme set $S'$, we have $\delta(S')\in[\lambda,(1+\epsilon)\lambda)$, which implies 
$$\delta_{H}(S')\le (1+\epsilon)p\cdot \delta(S')\le (1+\epsilon)^2p\lambda \le \frac{(1+\epsilon)^2}{(1-\epsilon)}\lambda_H\le 1.1\lambda_H.$$
(The first inequality is by \Cref{lem:karger-sample} and the second inequality by property of near-mincut extreme sets.)

When $\delta(c)<\epsilon\lambda$, let $G'$ be the graph formed by removing $c$ from $G$. For any Steiner cut $S'$ separating $S$, we have $\delta_{G'}(S')\ge \delta(S')-\delta(c)\ge (1-\epsilon)\lambda$. 
Use \Cref{lem:karger-sample} on $G'$ to get a subgraph $H$ 
with min-cut value $\lambda_H=O(\log n)$. Note that
$$\lambda_H\ge (1-\epsilon)p\lambda_{G'}\ge (1-\epsilon)^2 p\lambda.$$ Now, for any near-mincut extreme set 
$S''$, we have
\begin{align*}
    \delta_{H_1}(S'') &\le (1+\epsilon)p\cdot \delta_{G'}(S'') \le (1+\epsilon)p\cdot \delta(S'')\\
    &\le (1+\epsilon)^2p\lambda \le \frac{(1+\epsilon)^2}{(1-\epsilon)^2}\lambda_H\le 1.1\lambda_H.
\end{align*}
(The first inequality is by \Cref{lem:karger-sample}, the second inequality by the fact that $G'$ is a subgraph of $G$, and the third inequality by property of near-mincut extreme sets.)
\end{proof}

So far, we have constructed a subgraph $H$ of $G$ where every near-mincut extreme set has value at most $1.1\lambda_H$, where $\lambda_H$ is the Steiner connectivity of $S$ in $H$. We now pack a set of disjoint spanning trees in $H$. The next theorem follows from the work of Bang-Jensen {\em et al.}~\cite{Bang-JensenFJ95} and can also be derived from earlier work by Edmonds~\cite{Edmonds73}. We state a version of the theorem from \cite{ColeH03, BhalgatHKP07}. First, we need the following definition:
\begin{definition}
    Given a directed graph $G$ and a vertex $r$, a directionless tree rooted at $r$ is a (possibly non-spanning) tree of directed edges that is a subgraph of $G$, and where all edges incident $r$ are directed away from $r$. All other edges can have arbitrary direction.
\end{definition}

\begin{theorem}[\cite{ColeH03, BhalgatHKP07}]
\label{thm:tree-packing}
Given an Eulerian directed graph $G$, a root vertex $r$ and a value $C$, there exists $C$ edge-disjoint directionless trees rooted at $r$, such that the in--degree of every vertex $v\not= r$ in the union of all the trees is $\min\{\lambda(r,v), C\}$, where $\lambda(r,v)$ is the value of minimum $r$-$v$ cut. Such a tree packing can be obtained in $\tO(mC)$ time.
\end{theorem}

For an undirected graph, we can replace each undirected edge with two directed edges oriented in opposite direction and apply the above theorem to obtain the following corollary.

\begin{corollary}
\label{cor:tree-packing}
Given an undirected graph $G$, a root vertex $r$ and a value $C$, there exists $C$ (possibly non-spanning) trees rooted at $r$, such that (a) each vertex $v\ne r$ appears in at least $\min\{\lambda(r,v), C\}$ trees, and (b) every edge appears in at most two trees. Such a tree packing can be obtained in $\tO(mC)$ time.
\end{corollary}

Using this undirected tree packing, we can now reduce the problem of finding near-mincut extreme sets to finding extreme sets that $2$-respect a tree. We first define the $2$-respecting property.

\begin{definition}
    Given an undirected graph $G = (V, E)$ and a tree $T$ that is a subgraph of $G$ ($T$ may not be spanning), an extreme set $S\subsetneq V$ in $G$ is said to {\em $2$-respect} $T$ if there are at most two edges from the cut $(S, V\setminus S)$ that appear in $T$.
\end{definition}

Now, we are ready to further reduce near-mincut extreme sets to the following three problems:
\begin{itemize}
    \item Finding 2-respecting extreme sets: given a weighted undirected graph $G$ on vertices $S\cup \{c\}$, and a subgraph $T$ that is a tree spanning $S$ (it may or may not contain $c$), find a laminar family of vertex sets that contains all extreme sets in $G$ that $2$-respect $T$ and are subsets of $S$. 
    \item Merging two laminar trees: given a weighted undirected graph $G$, and two laminar families of vertex sets, merge these laminar families by selecting a single laminar collection of vertex sets from the two families that includes all extreme sets in $G$ that are in these families.
    \item Removing non-extreme sets from a laminar family: given a weighted undirected graph $G$ on vertices $S\cup \{c\}$, and a laminar family of vertex sets containing all near-mincut extreme sets of $G$ (but possibly other sets), find the near-mincut extreme sets of $G$ and discard the other sets that are not near-mincut extreme sets.
\end{itemize}

\begin{theorem}\label{thm:2-resp}
Suppose that given a weighted undirected graph $G$ on vertices $S\cup \{c\}$ containing $m$ edges, there are algorithms that can find 2-respecting extreme sets, merge two laminar trees, and remove non-extreme sets from a laminar family in $\tO(m)$ time. Then we can find whp all near-mincut extreme sets in $\tO(m)$ time. 
\end{theorem}


\begin{proof}
Given a near-mincut extreme sets problem instance in a graph $G$ on vertices $S\cup \{c\}$ where $\lambda$ denotes the Steiner connectivity of $S$, we first use \Cref{thm:sparsify-small-deg} to obtain a subgraph $H$. Let $\lambda'$ be the Steiner connectivity of $S$ in $H$. If we set $C = \lambda'$ in \Cref{cor:tree-packing} and apply it to $H$, then we get $\lambda' = O(\log n)$ trees spanning $S$. $S$ is spanned because for each $v\in S\setminus\{r\}$, $\lambda_H(r, v)\ge \lambda'$ by definition of Steiner connectivity, and $v$ appears in all $\lambda'=\min\{\lambda_H(r, v), C\}$ trees. We remark that these trees may or may not contain $c$.

Next, for each of these trees, we find a laminar family containing all $2$-respecting extreme sets using the first algorithm. We need to show that every near-mincut extreme set in $G$ will $2$-respect at least one of the trees, and therefore, will be in one of these laminar families. Set $\epsilon = 0.01$. For any $(1+\epsilon)\lambda$-weak extreme set $S'\subseteq S$ in $G$, we have that $S'$ is a $1.1$-approximate Steiner min-cut in $H$. Thus, the $\lambda'$ trees share at most $2.2\lambda'$ cut edges, since each edge appears at most twice in the trees by \Cref{cor:tree-packing}. So, on average, each tree has at most $2.2$ cut edges. 
Thus, there is at least one tree that has at most $2$ cut edges.

Now, we iteratively use the second algorithm to merge the laminar families returned for each tree into a single laminar family, and then remove the non-extreme sets from this family using the third algorithm to obtain the near-mincut extreme sets.

Next, we bound the running time. The application of \Cref{thm:sparsify-small-deg} takes $O(m)$ time, and that of \Cref{cor:tree-packing} makes $\tO(m\lambda')$ time. Since $\lambda' = O(\log n)$ by \Cref{thm:sparsify-small-deg}, we can conclude that the tree packing takes $\tO(m)$ time. Then, we run the extreme sets algorithm on each of the $O(\log n)$ trees, which takes $\tO(m)$ time. Since there are $O(\log n)$ trees, it follows that we need to call the merger algorithm $O(\log n)$ times, which takes $\tO(m)$ time. Finally, the algorithm to remove non-extreme sets a takes $O(m)$ time. Thus, the total runtime is $\tO(m)$.
\end{proof}

We will give the algorithm to find $2$-respecting extreme sets in the statement of \Cref{thm:2-resp} in \Cref{sec:tree2spider} and \Cref{sec:respect}, and the algorithm for merging two extreme sets trees into an extreme sets tree in \Cref{sec:merge}. Here, we give details of the last step, that of removing non-extreme sets from a laminar family.

\begin{lemma}
Given a laminar family containing all near-mincut extreme sets, we can remove all sets that are not near-mincut extreme sets in $O(m)$ time. 
\end{lemma}
\begin{proof}
First remove all sets $X$ with $\delta(X)\ge (1+\epsilon)\lambda$ or $X\setminus S\ne \emptyset$ because they cannot be near-mincut extreme sets. Then do a post-order traversal on the tree formed by the laminar family. When visiting some node $X$, compare the cut value of $X$ and all its children. If some child has cut value less or equal to $\delta(X)$, we remove $X$ from the family and assign its children to its parent in the tree. Given the cut values, the traversal takes $O(n)$ time.

Next we show that all cut values of sets in the laminar family can be computed in $O(m)$ time. For each vertex $u\in V$, let $p(u)$ be the collection of sets containing $u$ in the laminar family. Add set $V$ into the family, so that $p(u)$ is always nonempty. Because the family is laminar, $p(u)$ is a nested chain of sets. Let $l(u)$ be the minimal set in $p(u)$, then $p(u)$ is a path from the root to $l(u)$ in the laminar tree.
Every edge $(u, v)$ contributes to the cut values of sets separating $u$ and $v$, which are sets in exactly one of $p(u)$ or $p(v)$. On the laminar tree, they are on the path from $l(u)$ to $l(v)$ excluding the lowest common ancestor (LCA) of $l(u)$ and $l(v)$. Use Tarjan's offline LCA algorithm \cite{GabowT83} to calculate LCA$(l(u), l(v))$ of all edges $(u,v)$ in $O(m+n)$ time. Assign a label to each tree node. The labels are 0 initially. For every edge $(u, v)$ with weight $w$, add $w$ to the label of $l(u)$ and $l(v)$, and add $-2w$ to the label of LCA$(l(u), l(v))$. Then for every set in the family, its cut value is the sum of labels in the corresponding subtree. These sums of labels can be calculated in $O(n)$ time using dynamic programming.

Clearly, near-mincut extreme sets will not be removed by this algorithm. Next, we show that all sets that are not near-mincut extreme sets will indeed be removed. By the first step, we can only focus on non-extreme sets $X\subsetneq S$ that have cut value $\delta(X)<(1+\epsilon)\lambda$. For such a set $X$, there must be some extreme subset $Y\subsetneq X$ with $\delta(Y)\le \delta(X)<(1+\epsilon)\lambda$ (e.g., a vertex minimal subset of $X$ that violates the extreme condition for $X$ is a valid $Y$). Then $Y$ is a near-mincut extreme set, so $Y$ is in the family and has not been removed when visiting $X$ in the post-order traversal. Let $Y'$ be the ancestor of $Y$ that is also a child of $X$. Because the path from $Y$ to $Y'$ has survived the post-order traversal, the cut values will be monotone decreasing along this path. Thus, $\delta(Y')\le \delta(Y)\le \delta(X)$, which implies that $X$ will be removed.
\end{proof}

\section{Reduction to Spiders}
\label{sec:tree2spider}

In this section, we reduce the problem of finding $2$-respecting extreme sets in \Cref{thm:2-resp} to the special case when the tree $T$ is a \emph{spider}. This significantly simplifies the case analysis of the extreme sets algorithm. 

\begin{definition}[Spider]
A \emph{spider} is a rooted tree that is the edge-disjoint union of root-to-leaf paths.
\end{definition}

The full reduction has two steps. We first impose one additional restriction: we only need to find extreme sets for which the root $r$ of $T$ is on the path in $T$ between the two crossed edges. Of course, this requires the extreme set to cross \emph{exactly} two edges in $T$, and we call such a set \emph{exactly} $2$-respecting. The reduction is captured by the lemma below, which we prove in \Cref{sec:centroid} using the technique of centroid decomposition on a tree. 

\begin{lemma}\label{lem:red1}
Assume that given a weighted undirected graph, and a tree $T$ spanning all but at most one vertex, we can find in $\tilde{O}(m)$ time all extreme sets in $V(T)$ that either (a) $1$-respect $T$, or (b) exactly $2$-respect $T$ such that $r$ is on the path in $T$ between the two crossed edges. Then, we can find in $\tilde{O}(m)$ time all $2$-respecting extreme sets (with no additional condition).
\end{lemma}

Finally, we reduce this special case to one that assumes the tree $T$ is a spider. The lemma below is proved in \Cref{sec:spider-reduction} using the random branch contraction technique inspired by~\cite{Li19}.

\begin{lemma}\label{lem:red2}
Assume that given a weighted undirected graph, a special vertex $c$ and a spider $T$ spanning all but at most one vertex, we can find in $\tilde{O}(m)$ time a laminar family of $V(T)$ that includes all extreme sets that either (a) $1$-respect $T$, or (b) exactly $2$-respect $T$ such that $r$ is on the path in $T$ between the two crossed edges. Then, the same is true with ``spider'' replaced by a general ``tree''.
\end{lemma}

\subsection{Centroid Decomposition}\label{sec:centroid}

In this section, we prove \Cref{lem:red1}.

For a given tree $T$, the \emph{centroid} is a vertex $r$ such that if we root $T$ at $r$, then each subtree rooted at a vertex different from $r$ has at most half the total number of vertices. The centroid is guaranteed to exist for any tree, and one can be computed in linear time easily.

Root $T$ at the centroid $r$, and first call the $2$-respecting extreme sets algorithm under the special restriction described in the lemma statement. In particular, this algorithm returns a laminar family of subsets that includes all exactly $2$-extreme sets on $T$ for which $r$ is on the path in $T$ between the two crossed edges. 

Next, let $T_1,\ldots,T_\ell$ be the subtrees rooted at the children of $r$, with the additional edge between $r$ and the root of the subtree, so that $T_1,\ldots,T_\ell$ is an edge partition of $T$. We can split the set of subtrees into two groups such that each group has at most $2/3$ the total number of vertices. Without loss of generality, let $T_1,\ldots,T_k$ and $T_{k+1},\ldots,T_\ell$ be the two groups. The algorithm recursively solves two instances, one with all edges in $T_1,\ldots,T_k$ contracted to a single vertex, and one with all edges in $T_{k+1},\ldots,T_\ell$ contracted to a single vertex. Take the two laminar families returned by the recursive calls and ``uncontract'' the contracted vertex in any set that contains it, i.e., replace it with the vertices in $T_1\cup\cdots\cup T_k$ or $T_{k+1}\cup\cdots\cup T_\ell$ depending on which instance. This does not destroy laminarity of the two families. We then use \Cref{lem:merge} to merge the three laminar families found overall (including the one from the non-recursive case above). 

We claim that the resulting laminar family includes all extreme sets $2$-respecting $T$. There are a few cases:
 \begin{enumerate}
 \item If an extreme set $1$-respects $T$, then it is picked up by the non-recursive case.
 \item If an extreme set crosses two edges, one in $T_1\cup\cdots\cup T_k$ and one in $T_{k+1}\cup\cdots\cup T_\ell$, then it satisfies the specific condition that the root $r$ is on the path in $T$ between the two crossed edges, so the non-recursive case outputs this set.
 \item If an extreme set crosses two edges, either both in $T_1\cup\cdots\cup T_k$ or both in $T_{k+1}\cup\cdots\cup T_\ell$, then it survives when the other ($T_1\cup\cdots\cup T_k$ or $T_{k+1}\cup\cdots\cup T_\ell$) is contracted, so it is output by the corresponding recursive algorithm.
 \end{enumerate}
It follows that all $2$-respecting extreme sets are output by the algorithm. By \Cref{lem:merge}, they all survive the merging step, and are therefore included in the final output.

As for running time, the recursion depth is $O(\log n)$ since the number of vertices in $T$ drops by a constant factor on each recursive call. Also, on each recursion level, the sum of the sizes of the instances is $O(m+n\log n)$ by the following argument. Each vertex in the original tree $T$ appears in at most one instance (as a non-contracted vertex), and each instance has an additional $O(\log n)$ contracted vertices (one from each recursive call before it) and possibly one vertex $c$ in $V$ not spanned by $T$. There are at most $m$ edges across the instances whose endpoints are not $c$ nor contracted vertices, since each original edge appears in at most one instance (the one containing both of its endpoints as non-contracted vertices, if any). Each instance with $k$ vertices also gets an extra $O(k\log n)$ edges adjacent to either one of the $O(\log n)$ contracted vertices or $c$. It is not hard to see that the total number of vertices among the instances is $O(n)$, so this is an additional $O(n\log n)$ edges. It follows that the sum of the sizes of the instances at each level is $O(m+n\log n)$, and over all the $O(\log n)$ levels, this is still $\tilde{O}(m)$.
 
\subsection{Reduction from Trees to Spiders}\label{sec:spider-reduction}

In this section, we further simplify to the case when $T$ is a spider, proving \Cref{lem:red2}. The idea is simple: we compute a heavy-light decomposition of the tree, viewed as a set of edge-disjoint branches, and randomly contract a subset of them so that the remaining graph is a spider. We ensure that for any two fixed edges for which the root is on the path between them, with probability $\Omega(1/\log^2n)$ both edges survive the contraction.

\begin{figure}[t]
\centering
\includegraphics[width=.4\textwidth]{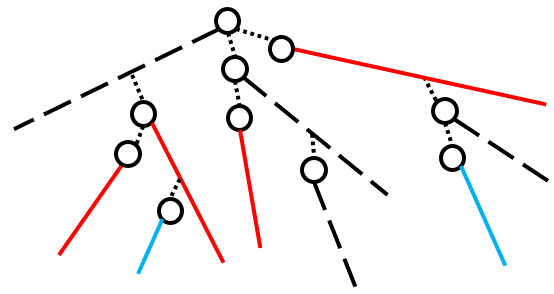}
\caption{Build spider from heavy-light decomposition. Dashed branches: not sampled. Red branches: sampled and used in spider. Blue branches: sampled but discarded.}\label{fig:spider}
\end{figure}

More precisely, we define a heavy-light decomposition as a partition $\mathcal P$ of the edges of $T$ into monotone paths (i.e., consecutive vertices along the path have increasing/decreasing distance from the root) called \emph{branches}, such that for any vertex $v$ in $T$, the path from $v$ to the root $r$ shares edges with $O(\log n)$ many branches. The algorithm samples each branch in $\mathcal P$ independently with probability $1/\log n$, and we keep all sampled branches whose path from (any vertex on) the branch to the root does not intersect any edge of another sampled branch; see Figure~\ref{fig:spider}. The algorithm contracts all other branches. It repeats this process $O(\log^3n)$ times, and for each instance, it calls the extreme sets algorithm on a spider described in the statement of \Cref{lem:red2}. The algorithm then ``uncontracts'' all edges to obtain collections of sets of vertices in $G$, and merges them using \Cref{lem:merge}. We claim that this algorithm correctly outputs all extreme sets promised by \Cref{lem:red2}.

We first claim that the resulting graph is indeed a spider. Indeed, for every branch $B$ that is kept, consider the path from the branch to the root; any other branch sharing edges with this path was not sampled, otherwise branch $B$ would not be kept. It follows that the branch hangs off the root in the contracted graph. Since all branches are monotone and hang off the root, the contracted graph must be a spider.

Finally, we claim that for any two edges of $T$ for which the root is on the path between them, with probability $\Omega(1/\log^2n)$ both edges survive the contraction, i.e., the branches containing them are kept. For a single edge $e$, in order for its respective branch $B$ to be kept, that branch must be sampled and none of the $O(\log n)$ other branches sharing edges with the path $P_e$ from $e$ to the root can be sampled. This occurs with probability $1/\log n \cdot (1-1/\log n)^{O(\log n)}=\Omega(1/\log n)$. By assumption, for the two edges $e_1,e_2$, the paths $P_{e_1}$ and $P_{e_2}$ are edge-disjoint and connected at the root, so the set of branches sharing edges with $P_{e_1}$ is disjoint from the set of branches sharing edges with $P_{e_2}$. It follows that the event that $e_1$ survives the contraction is independent from the event for $e_2$, and the overall probability of success is $\Omega(1/\log^2n)$.

Therefore, if we repeat this procedure $O(\log^3n)$ times, then with high probability, for any two such edges $e_1,e_2$, they both survive in one of the resulting spiders. In particular, if there is an exactly $2$-respecting extreme set crossing $e_1$ and $e_2$, then that extreme set survives the contraction as well. Likewise, a $1$-respecting extreme set crossing $e_1$ or $e_2$ survives as well. It follows that the extreme sets algorithm on a spider outputs the contracted version of this extreme set. The set is then uncontracted to the original extreme set, and then included in the final output after merging. It follows that with high probability, all targeted extreme sets are output by the algorithm.

\section{2-respecting Extreme Sets on a Spider}
\label{sec:respect}

In this section, we propose an efficient algorithm that, given a weighted undirected graph on vertices $S\cup \{c\}$ and a tree $T$ spanning $S$, finds all extreme sets in $S$ that 2-respect $T$. Using the reduction in \Cref{sec:tree2spider}, we can assume that $T$ is a spider. Such extreme sets can be divided into four `universes': one subtree, complement of one subtree, two subtrees and complement of two subtrees. We design algorithms to find extreme sets in each universe, and merge all the families by \Cref{lem:merge}. 

We now introduce some notations exclusive to this section. For a tree $T$, define $u^\downarrow\subseteq V(T)$ as the vertices in the subtree of $T$ rooted at $u$, and $u^\uparrow\subseteq V(T)$ as the vertices on the path from $u$ to the root.
The complement $\overline{X}=V(T)\setminus X$ is defined with respect to vertices on the tree. We say that two vertices $u,v\in V(T)$ are \emph{incomparable} if $u^\downarrow\cap v^\downarrow=\emptyset$, i.e., neither is an ancestor or descendant of the other, and we sometimes use the notation $u\bot v$ to indicate that $u$ and $v$ are incomparable. Likewise, we say that $u,v$ are \emph{comparable} if $u^\downarrow\cap v^\downarrow\ne\emptyset$, and we sometimes use the notation $u\parallel v$. Note that on a spider, two non-root vertices are incomparable iff they lie on different branches, and they are comparable iff they lie on the same branch.

\subsection{Universe 1: One Subtree}
The one subtree case is simple. Let $\mathcal F$ be the laminar family of all (vertex sets of) subtrees of $T$: $\mathcal F=\{v^\downarrow:v\in V(T)\}$, then $\mathcal{F}$ trivially contains all extreme sets in the form of one subtree. 

\subsection{Universe 2: Complement of One Subtree}
Note that all sets in this universe contain the root, so any laminar family of sets in this universe must be a nested chain, which means the cut edges of the sets on the tree must lie on the same branch. We can actually find this main branch.
\begin{lemma}
\label{lem:1-respect-proof}
Let $S_1$ be the set with minimum cut value among all subtrees and complement of subtrees. When $S_1$ is a subtree, let $S_1=u_1^\downarrow$, otherwise let $S_1=\overline{u_1^\downarrow}$. (They are not equivalent when $c\notin V(T)$.) If $S=\overline{u^\downarrow}$ is extreme, then $u\parallel u_1$.
\end{lemma}
\begin{proof}
Assume for contradiction that $u\perp u_1$. Then $u_1^\downarrow\subsetneq S$ and $\delta(u_1^\downarrow)>\delta(S)$ because $S$ is extreme. When $S_1=u_1^\downarrow$ this contradicts $S_1$'s minimality. When $S_1=\overline{u_1^\downarrow}$, by posi-modularity \[\delta(S_1)+\delta(S)\ge \delta(S\setminus S_1)+\delta(S_1\setminus S)=\delta(u_1^\downarrow)+\delta(u^\downarrow)>\delta(S)+\delta(S_1),\] contradiction.
Therefore $u\parallel u_1$.
\end{proof}

It immediately follows that $\{\overline{u^\downarrow}:u\parallel u_1\}$ is a laminar family containing all extreme sets in the form of complement of one subtree.

\subsection{Universe 3: Two Subtrees}
\label{sec:two-subtrees}
In this section, we compute a laminar family of subsets such that each extreme set composed of the union of two subtrees is included in the family, i.e., they can be written as $u^\downarrow\cup v^\downarrow$ for some $u,v$ on different branches of the spider. We introduce two concepts central to the algorithm: \emph{partners} and \emph{bottlenecks}.

\paragraph{Partners.}

Informally, we consider a vertex $v$ to be a vertex $u$'s \emph{partner} if $u^\downarrow\cup v^\downarrow$ is a potential extreme set. A necessary condition for this to happen is 
\BG
\delta(v^\downarrow)>\delta(u^\downarrow\cup v^\downarrow) \iff \delta(u^\downarrow, v^\downarrow)>\frac 12 \delta(u^\downarrow)\tag{P}\label{eq:P}
\EG
Note that for a fixed $u$ there cannot be two incomparable vertices $v$ satisfying condition~(\ref{eq:P}). Therefore, the partners of $u$ are pairwise comparable (if they exist), so on a spider, they must lie on a single branch of the spider, and we can define the {\it lowest partner} $p(u)$ to be the partner of $u$ of highest depth in the tree (i.e., farthest away from the root). We also require $p(u)\perp u$ because we assume $v$ in on a different branch with $u$, and we say $p(u)$ does not exist if there is no vertex $v$ satisfying (\ref{eq:P}), or equivalently, the lowest partner is comparable to $u$.

\begin{fact}\label{fact:extreme}
If $u^\downarrow\cup v^\downarrow$ is extreme, then (\ref{eq:P}) holds, and $v\in p(u)^\uparrow$ and $u\in p(v)^\uparrow$.
\end{fact}

We now show that we can efficiently compute $p(u)$ for every vertex $u$.
\begin{lemma}\label{lem:compute-partner}
We can compute $p(u)$ for every $u\in V(T)-r$ in $\tilde{O}(m)$ time.
\end{lemma}
\begin{proof}
We give an algorithm computing $p(u)$ for a branch $B$ in time proportional to (up to polylogarithmic factors) $|B|$ plus the number of edges incident to vertices in $B$. Repeating this algorithm for all branches gives an $\tilde{O}(m)$ time algorithm.

Iterate over all $u\in B$ from the leaf upwards. This means that in each iteration, we add a new node into $u^\downarrow$. We use a heap to maintain value $\delta(u^\downarrow, B')$ for every other branch $B'\ne B$. (Recall that each branch $B'$ is a root-to-leaf path minus the root.) Also, for each branch $B'\ne B$, we maintain a sorted list of added edges $(u,v)$ where $u\in B,v\in B'$, sorted by the position of $v$ in branch $B'$ from leaf to root.

 In each iteration with new node $u$, for every edge $(u,v)$ incident on $u$ with $v\notin B\cup\{r\}$, add its weight to the value at the branch $B$ containing $v$, and also insert edge $(u,v)$ to the sorted list of edges for $B'$. After the update step, we query the branch $B'\ne B$ with maximum value $\delta(u^\downarrow, B')$. If $\delta(u^\downarrow, B')> \frac 12 \delta(u^\downarrow)$, then we find the lowest vertex $v\in B$ satisfying $\delta(u^\downarrow, v^\downarrow)>\frac 12 \delta(u^\downarrow)$, which can be done by binary searching over $v$ and taking a prefix sum of the sorted list to determine each $\delta(u^\downarrow,v^\downarrow)$. We set $p(u)$ to be this vertex $v$.
\end{proof}

Recall that by definition, a partner $v$ should be incomparable to $u$ and satisfy condition (\ref{eq:P}). If $p(u)$ does not exist, $u$ cannot be one of the two subtrees that form an extreme set. Therefore, after computing $p(u)$ for all $u$, we can contract every $u$ whose lowest partner does not exist to its parent without losing any extreme set composed of two subtrees. After this preprocessing step, we can
assume the lowest partner $p(u)$ exists for all $u\in V(T)$.

\paragraph{Bottlenecks.}
We now define the concepts of \emph{weak bottleneck} and \emph{bottleneck} as a sort of upper bound on the cut size of an extreme set. The \emph{weak bottleneck} for a vertex $u$
is defined as $b_{weak}(u)=\min_{w\in u^\downarrow -u}\delta(u^\downarrow\setminus w^\downarrow)$, and the bottleneck is 
\[ b(u) = \min_{v\in u^\downarrow}b_{weak}(v) = \min_{v\in u^\downarrow, w\in v^\downarrow-v} \delta(v^\downarrow\setminus w^\downarrow) .\]
The fact below explains the motivation of bottleneck as an upper bound.

\begin{fact}
If $u^\downarrow\cup v^\downarrow$ is extreme, then $\delta(u^\downarrow\cup v^\downarrow)<\min\{b(u), b(v)\}$.
\end{fact}
\begin{proof}
By the definition of bottleneck, there exists some $w_1\in u^\downarrow$ and $w_2\in w_1^\downarrow-w_1$ such that  $b(u)=\delta(w_1^\downarrow\setminus w_2^\downarrow)$. Also $\delta(u^\downarrow\cup v^\downarrow) < \delta(w_1^\downarrow\setminus w_2^\downarrow)$ because $\delta(w_1^\downarrow\setminus w_2^\downarrow)\subseteq u^\downarrow$ and $u^\downarrow\cup v^\downarrow$ is extreme. Therefore, $\delta(S)<b(u)$. Swapping $u$ and $v$ in the argument gives $\delta(S)<b(v)$ as well.
\end{proof}

The next fact establishes monotonicity of bottleneck, which is useful for a binary search procedure we execute later on.
\begin{fact}
\label{lem:bottleneck-monotonic}
$b(u)$ is monotonic decreasing in a branch from leaf to root.
\end{fact}

\begin{fact}
We can compute $b(u)$ for every $u\in V(T)-r$ in $\tilde{O}(m)$ time.
\end{fact}

\begin{proof}
Note that $b(u)$ can be computed independently for each branch, so we focus on a single branch $B$. We first calculate $b_{weak}(u)$ for every $u$. Observe that $\delta(u^\downarrow\setminus w^\downarrow)=2\delta(u^\downarrow\setminus w^\downarrow,w^\downarrow)+\delta(u^\downarrow)-\delta(w^\downarrow)$. We can easily calculate $\delta(w^\downarrow)$ for all $w$ by traversing the vertices of the branch from the leaf upwards, and using that for a parent $v$ of vertex $w$, we have $\delta(v^\downarrow)=\delta(w^\downarrow)+\delta(v)-2\delta(v,w^\downarrow)$. This takes time proportional to $|B|$ plus the number of edges incident to vertices in $B$. Next, initialize a dynamic array with value $-\delta(w^\downarrow)$ on each vertex $w$. Traverse the branch from the leaf upwards, and for the current vertex $u$, we take all edges $(u,v)$ for $v\in u^\downarrow$, and for each such edge, we add twice its weight to all vertices on the array from $v$ inclusive to $u$ exclusive. This way, each vertex $w\in u^\downarrow$ has current value $2w(u^\downarrow\setminus w^\downarrow, w^\downarrow)-\delta(w^\downarrow)$, so we can query the minimum value of the prefix of the array up to $u$ to obtain $\min_{w\in u^\downarrow}2w(u^\downarrow\setminus w^\downarrow, w^\downarrow)-\delta(w^\downarrow)$. Finally, adding $\delta(u^\downarrow)$ to the query gives us $\min_{w\in u^\downarrow}\big(2\delta(u^\downarrow\setminus w^\downarrow,w^\downarrow)+\delta(u^\downarrow)-\delta(w^\downarrow)\big) = b_{weak}(u)$. Altogether, the algorithm on branch $B$ takes time proportional to (up to logarithmic factors) $|B|$ plus the number of edges incident to vertices in $B$. Summed over all branches $B$, this is $\tilde O(m)$ time total.
\end{proof}

\subsubsection{Lowest Partner Condition}

The follow lemma captures the key property of our definition of the lowest partner $p(u)$.

\begin{lemma}
\label{lem:subtree-partner-condition}
If $u^\downarrow\cup v^\downarrow$ is extreme, then for every $ w\in u^\downarrow$ whose lowest partner exists, we have $p(w)\parallel p(u)$. Symmetrically, for all $w\in v^\downarrow$, we have $p(w)\parallel p(v)$.
\end{lemma}

\begin{figure}[t]
\centering
\includegraphics[width=.3\textwidth]{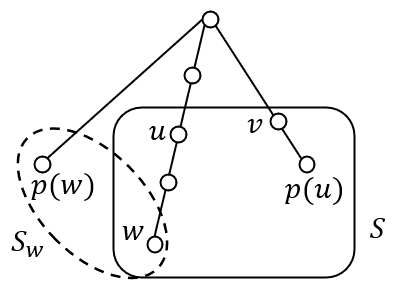}
\caption{Bad case in the proof of \Cref{lem:subtree-partner-condition}}
\end{figure}

\begin{proof}
Assume for contradiction that there exists some $w\in u^\downarrow$ with $p(w)\perp p(u)$. Since (\ref{eq:P}) holds for $u$ and $v$, the lowest partner $p(u)$ must be lower than $v$, and in particular, they share the same branch of the spider, so $p(w)\perp v$. By definition of lowest partner, we must have $p(w)\perp w$, and since $u$ and $w$ share a branch, this implies $p(w)\perp u$.
It follows that $p(w)^\downarrow\cap S=\emptyset$.
Let $S_w = w^\downarrow\cup   p(w)^\downarrow$. By condition~(\ref{eq:P}), we have $\delta(S_w) < \delta(p(w)^\downarrow)$. Since $S$ is extreme, $\delta(S) < \delta(S\setminus w^\downarrow)$. Adding these two inequalities contradicts $\delta(S)+\delta(S_w)\ge \delta(S\setminus S_w)+\delta(S_w\setminus S)=\delta(S\setminus w^\downarrow)+\delta(p(w)^\downarrow)$, which holds by posi-modularity.
\end{proof}

This lemma allows us to pair up branches as follows. Compute lowest partners $p(u)$ for all vertices $u$. Then, for each branch $B$, take the lowest vertex $u$ in that branch whose lowest partner $p(u)$ is defined (if it exists), and let $f(B)$ be the branch containing $p(u)$. We pair up branches $B,B'$ satisfying $B'=f(B)$ and $B=f(B')$. Some branches may not be paired; we leave them alone.

\begin{lemma}
For any extreme set $u^\uparrow\cup v^\uparrow$, the two branches containing $u$ and $v$ are paired up.
\end{lemma}
\begin{proof}
By \Cref{fact:extreme} and \Cref{lem:subtree-partner-condition}, for an extreme set $u^\uparrow\cup v^\uparrow$, both $p(u)$ and $p(v)$ are defined, and for the lowest vertices $u'\parallel u$ and $v'\parallel v$ whose $p(u'),p(v')$ are defined, we have $p(u')\parallel p(u)\parallel v$ and $p(v')\parallel p(v)\parallel u$. In other words, the two branches containing $u$ and $v$ are paired up, as needed.
\end{proof}

Therefore, we can process each pair of branches separately by contracting all other branches to the root. The remaining task is to compute, for each pair of branches $B,B'$, a laminar family that contains all extreme sets of the form $u^\downarrow\cup v^\downarrow$. The laminar family we construct is
\[ \mathcal F_{(B,B')} = \big\{ u^\downarrow\cup v^\downarrow: u\in B,\,v\in B',\,\delta(u^\downarrow\cup v^\downarrow)<\min\{b(u),b(v)\}\big \} .\]

\begin{lemma}\label{lem:F-laminar}
The set $\mathcal F_{(B,B')}$ is laminar. That is, any two sets $u_1^\uparrow\cup v_1^\uparrow,\,u_2^\uparrow\cup v_2^\uparrow\in\mathcal F_{(B,B')}$ satisfy either $u_1\in u_2^\downarrow,\,v_1\in v_2^\downarrow$ or $u_2\in u_1^\downarrow,\,v_2\in v_1^\downarrow$.
\end{lemma}
\begin{proof}
Suppose for contradiction that $u_1\in u_2^\downarrow-u_2$ and $v_2\in v_1^\downarrow-v_1$ (without loss of generality). Let $S_1=u_1^\uparrow\cup v_1^\uparrow$ and $S_2=u_2^\uparrow\cup v_2^\uparrow$. Then, the sets $S_1$ and $S_2$ cross, and by posi-modularity,
\begin{align*}
 \delta(S_1)+\delta(S_2) &\ge \delta(S_1\setminus S_2)+\delta(S_2\setminus S_1)\\
 &=\delta(v_1^\downarrow\setminus v_2^\downarrow)+\delta(u_2^\downarrow\setminus u_1^\downarrow)\\
 &\ge b(v_1)+b(u_2).
\end{align*}
But $S_1,S_2\in\mathcal F_{(B,B')}$ implies that $\delta(S_1)<b(v_1)$ and $\delta(S_2)<b(v_2)$, a contradiction.
\end{proof}
\begin{lemma}
Over all branches $B,B'$, we can compute all pairs $(u,v):u\in B,v\in B'$ with $u^\downarrow\cup v^\downarrow\in\mathcal F$ in $\tilde{O}(m)$ time total.
\end{lemma}
\begin{proof}
For a fixed pair of branches $B,B'$, we describe an algorithm that finds all $u^\downarrow\cup v^\downarrow\in\mathcal F_{(B,B')}$ such that $b(u)\le b(v)$. The other case $b(u)>b(v)$ can be handled by swapping $B$ and $B'$ and running the same algorithm. Repeating the algorithm for all pairs of branches establishes the lemma.

Fix a pair of branches $B, B'$. We maintain a range minimum query data structure $\mathcal D$ on the vertices in branch $B'$. Initialize the data structure with value $\delta(v^\downarrow)$ for each vertex $v\in B'$. 

Now iterate through the vertex $u\in B$ from leaf to root.  Let the current iteration be at vertex $u\in B$. First, for each edge $(u,v)$ with $v\in B'$, subtract twice its weight from all vertices in $v^\uparrow$ in the data structures, which is an interval update. This ensures that each element $v$ has current value $\delta(v^\downarrow)-2\delta(u^\downarrow,v^\downarrow)$ in the data structure. Next, we seek all sets $u^\downarrow\cup v^\downarrow\in\mathcal F$ for the current $u$, assuming $b(u)\le b(v)$. By monotonicity of $b(v)$ (\Cref{lem:bottleneck-monotonic}), the vertices $v\in B'$ satisfying $b(u)\le b(v)$ form a consecutive interval $I$ in the branch which can be found by binary search. To find vertices $v\in I$ with $u^\downarrow\cup v^\downarrow\in\mathcal F$ and $b(u)\le b(v)$, we are looking for vertices $v\in I$ satisfying $\delta(u^\downarrow\cup v^\downarrow)<b(u)$. Note that $\delta(u^\downarrow\cup v^\downarrow)=\delta(u^\downarrow)+\delta(v^\downarrow)-2\delta(u^\downarrow,v^\downarrow)$, so this is equivalent to $\delta(v^\downarrow)-2\delta(u^\downarrow,v^\downarrow) < b(u)-\delta(u^\downarrow)$, so it suffices to find all vertices $v$ whose value in $\mathcal D$ is less than $b(u)-\delta(u^\downarrow)$, a value independent of $v$. This can be done by repeatedly querying for the vertex of minimum value inside interval $I$ in $\mathcal D$, and if the value is less than $b(u)-\delta(u^\downarrow)$, then add a large value $M$ to the value of $v$ and repeat, ensuring a different vertex has the minimum value this time; this recovers all such $(u,v)$, and we can subtract $M$ from these vertices $v$ once we are done. 

Altogether, for vertex $u\in B$, the total running time is proportional to (up to $\textup{polylog}(n)$ factors) the number of edges $(u,v)$ with $v\in B'$ plus the number of pairs $(u,v)$ found. The former totals at most the number of edges between branches $B$ and $B'$, and the latter totals $O(|B|+|B'|)$ since $\mathcal F$ is a laminar family by \Cref{lem:F-laminar}. Finally, over all pairs of branches $B,B'$, the number of edges between pairs of branches totals at most $m$, and the sum of $O(|B|+|B'|)$ totals $O(n)$. It follows that the entire algorithm takes $\tilde{O}(m)$ time.
\end{proof}

Next, note that the laminar families $\mathcal F_{(B,B')}$ are disjoint from each other since they are contained in their respective branches $B\cup B'$ which are pairwise disjoint. It follows that their union $\bigcup_{(B,B')}\mathcal F_{(B,B')}$ is also a laminar family. We have thus computed a laminar family containing all desired extreme sets in $\tilde{O}(m)$ time.

\subsection{Universe 4: Complement of Two Subtrees}
This section finds a laminar family containing all extreme sets in the form of complement of two subtrees.
The algorithm has the same spirit as in two subtrees case.

\subsubsection{Find Main Branch}
All sets of the form $\overline{u^\downarrow\cup v^\downarrow}$ contain the root, so a laminar sub-family must be a nested chain of sets. This means the cut edges on the tree must be contained in two branches. We first locate one of the two branches to be some $u_0^\uparrow$. Then the problem can be reduced to finding extreme sets in the form of $\overline{u^\downarrow\cup v^\downarrow}$ where $u\in u_0^\uparrow$.

\begin{lemma}
Let $S_1$ be the set with minimum cut value among four types of sets: a subtree, complement of a subtree, two subtrees, and complement of two subtrees.
Describe set $S_1$ by $u_1^\downarrow$, $\overline{u_1^\downarrow}$, $u_1^\downarrow\cup v_1^\downarrow$ or $\overline{u_1^\downarrow\cup v_1^\downarrow}$ respectively in the four cases.

Any extreme set of the form $\overline{u^\downarrow\cup v^\downarrow}$ has one endpoint in the branch of $u_1$ in the first two cases, and has one endpoint in the branch of $u_1$ or $v_1$ in the last two cases.
\end{lemma}

\begin{proof}
Let $S=\overline{u^\downarrow\cup v^\downarrow}$ be any such extreme set. Let $X=u_1^\downarrow$ in the first two cases, and $X=u_1^\downarrow\cup v_1^\downarrow$ in the last two cases, so that either $S_1=X$ or $S_1=\overline{X}$.

Assume for contradiction that neither $u$  or $v$ is comparable to $u_1$ in the first two cases, and to $u_1$ or $v_1$ in the last two cases, which means $X\subsetneq S$ and $u_1^\downarrow\cup v_1^\downarrow\subsetneq \overline{X}$. Since $S$ is an extreme set, this means that $\delta(X)>\delta(S)$. When $S_1=X$, we obtain $\delta(S_1)>\delta(S)$, which contradicts the minimality of $\delta(S_1)$. When $S_1=\overline{X}$, by posi-modularity \[\delta(S)+\delta(S_1)\ge \delta(S\setminus S_1)+\delta(S_1\setminus S)=\delta(u^\downarrow\cup v^\downarrow)+\delta(X)>\delta(S_1)+\delta(S),\] contradiction. Therefore $u$ or $v$ is comparable to $u_1$ in the first two cases, and to $u_1$ or $v_1$ in the last two cases.
\end{proof}

In the first two cases, we fix the main branch containing $u_1$, which is $u_0^\uparrow$ where $u_0$ is the leaf of that branch. In the last two cases, we try fixing main branches $u_1$ and $v_1$, compute the two laminar families, and merge them using \Cref{lem:merge}. From now on, assume that we have correctly identified the branch $u_0^\uparrow$.

If the tree $T$ spans all but one vertex $c$, then we attach $c$ below $u_0$ in the tree. This way, the new tree $T'$ now spans all vertices, and all extreme sets we wish to find (in particular, they do not include $c$) are still of the form $\overline{u^\downarrow\cup v^\downarrow}$.

\subsubsection{Partner Condition}
Like in two subtrees case, the idea is to restrict the potential partners onto a path, but with a different partner condition.
This time, we define the lowest partner
\BG
p(u)=\arg\min_{v\perp u} \delta(\overline{u^\downarrow\cup v^\downarrow}, v^\downarrow) \tag{P\textsc{\char13}}\label{eq:P'}
\EG
Since $\delta(\overline{u^\downarrow\cup v^\downarrow}, v^\downarrow)=\frac 12(\delta(u^\downarrow\cup v^\downarrow)+\delta(v^\downarrow)-\delta(u^\downarrow))$, lowest partners can be calculated in the same way as in \Cref{lem:compute-partner} from the two subtrees case.

\begin{lemma}
\label{lem:universe2-partner}
If $S=\overline{u^\downarrow\cup v^\downarrow}$ is an extreme set, then $v\in p(u)^\uparrow$.
\end{lemma}
\begin{proof}
Assume for contradiction that $v\notin p(u)^\uparrow$, so that either $v\perp p(u)$ or $v\in p(u)^\downarrow-p(u)$. There are two cases:

\paragraph{\rm{Case 1: $v\perp p(u)$.}} $p(u)$ is also incomparable to $u$ by definition, so $p(u)\subsetneq S$. Because $S$ is extreme, $\delta(p(u)^\downarrow)>\delta(S)$, which implies
\[\delta(p(u)^\downarrow,    S-p(u)^\downarrow)> \delta(S-p(u)^\downarrow, u^\downarrow\cup v^\downarrow)\ge \delta(S-p(u)^\downarrow, v^\downarrow).\]
Adding $\delta(v^\downarrow, p(u)^\downarrow)$ to both sides  gives $\delta(\overline{u^\downarrow\cup p(u)^\downarrow}, p(u)^\downarrow) > \delta(S, v^\downarrow)=\delta(\overline{u^\downarrow\cup v^\downarrow},v^\downarrow)$, which contradicts minimality in (\ref{eq:P'}).

\paragraph{\rm{Case 2: $v\in p(u)^\downarrow-p(u)$.}} Let $X=\overline{u^\downarrow\cup p(u)^\downarrow}\subsetneq S$ and $Y=S\setminus X=p(u)^\downarrow\setminus v^\downarrow$. Because $S$ is extreme,
\[\delta(X)>\delta(S)\implies \delta(X,Y)>\delta(Y,u^\downarrow\cup v^\downarrow)\ge \delta(Y,v^\downarrow).\]
Adding $\delta(X,v^\downarrow)$ to both sides gives $\delta(X,p(u)^\downarrow)>\delta(S,v^\downarrow)$, or in other words, $\delta(\overline{u^\downarrow\cup p(u)^\downarrow}, p(u)^\downarrow)>\delta(\overline{u^\downarrow\cup v^\downarrow},v^\downarrow)$, which contradicts minimality in (\ref{eq:P'}).
\end{proof}

\subsubsection{Find the Second Branch}
We would now like to identify a second branch to locate all extreme sets. Our key observation is that if $S=\overline{u^\downarrow\cup v^\downarrow}$ is extreme, then for any $w$ that is incomparable to both $u$ and $v$, $\delta(w^\downarrow)>\delta(S)$ because $w^\downarrow\subsetneq S$. We call this the \emph{subtree cut condition}:
\BG
\forall w,\, w\bot u,\, w\bot v: \delta(w^\downarrow)>\delta(S) \tag{S}\label{eq:S}
\EG
Therefore, $\delta(S)$ is less than the minimum subtree cut in all branches other than $u$'s and $v$'s (or equivalently, $p(u)$'s by \Cref{lem:universe2-partner}).

Next, define the \emph{optimal partner} $opt(u)=\arg\min_{v\in p(u)^\uparrow}\delta(u^\downarrow\cup v^\downarrow)$. We only calculate the optimal partners for $u\in u_0^\uparrow$, which can be done in $\tilde{O}(m)$ time.
For each branch, calculate the minimum cut value among all subtrees on the branch. List the values as a sequence to perform range minimum queries.

\begin{lemma}
\label{lem:second-branch}
Let $u_1$ be the highest node on main branch $u_0^\uparrow$ such that $\overline{u_1^\downarrow\cup opt(u_1)^\downarrow}$ satisfies subtree cut condition (\ref{eq:S}). Then, any extreme set $S=\overline{u^\downarrow\cup v^\downarrow}$ with $u\in u_0^\uparrow$ has $v\parallel opt(u_1)$.
\end{lemma}

\begin{figure}[t]
\centering
\includegraphics[width=.6\textwidth]{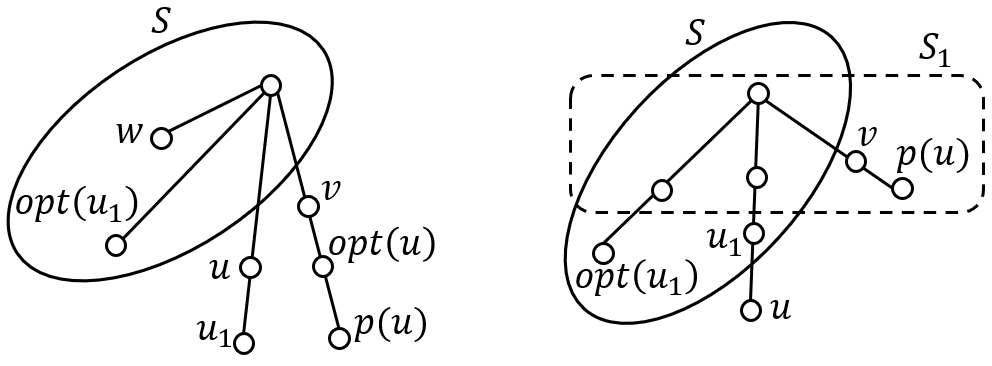}
\caption{Proof of \Cref{lem:second-branch}. Left: the second case where $u$ is above $u_1$. Right: the third case where $u$ is below $u_1$.}
\end{figure}
\begin{proof}
For any extreme set $S=\overline{u^\downarrow\cup v^\downarrow}$ with $u\in u_0^\uparrow$, we define $u_1$ as in the lemma. Assume for contradiction that $v\perp opt(u_1)$. We case on the location of $u$: either $u=u_1$, or $u\in u_1^\uparrow-u_1$, or $u\in u_1^\downarrow-u_1$. Note that $u\parallel u_1$ because both are in $u_0^\uparrow$.

First, suppose that $u=u_1$. Then $opt(u_1)=opt(u)$, which is in $p(u)^\uparrow$ by definition of $opt(u)$. Vertex $v$, as a partner of $u$, is also in $p(u)^\uparrow$ by \Cref{lem:universe2-partner}. This contradicts $v\perp opt(u_1)$.

Second, suppose that $u\in u_1^\uparrow-u_1$. By definition of $u_1$, we have that $\overline{u^\downarrow\cup opt(u)^\downarrow}$ does not satisfy subtree cut condition (\ref{eq:S}), since otherwise $u$ would be a better choice than $u_1$. Since $\overline{u^\downarrow\cup opt(u)^\downarrow}$ does not satisfy (\ref{eq:S}), there exists some $w$ incomparable to both $u$ and $p(u)$ such that $\delta(w^\downarrow) \le \delta(u^\downarrow\cup opt(u)^\downarrow)$. By definition of $opt(u)$, we have $ \delta(u^\downarrow\cup opt(u)^\downarrow)\le\delta(S)$. These two inequalities implies $\delta(w^\downarrow)\le \delta(S)$. However, since $v\parallel p(u)$ by \Cref{lem:universe2-partner} and $w$ is incomparable to both $u$ and $p(u)$, we have $w^\downarrow\subsetneq S$, and since $S$ is extreme, this implies that $\delta(w^\downarrow)>\delta(S)$, a contradiction.

The final case is $u\in u_1^\downarrow-u_1$. Let $S_1=\overline{u_1^\downarrow\cup opt(u_1)^\downarrow}$. The sets $S_1$ and $S$ cross because $u\parallel u_1$ and $v\perp opt(u_1)$. Since $S$ is extreme, we have $\delta(S\setminus S_1)>\delta(S)$, so by posi-modularity, 
we have $\delta(S_1\setminus S)<\delta(S_1)$. Notice that $S_1\setminus S=v^\downarrow$. It follows that $\delta(v^\uparrow)<\delta(S_1)$, which contradicts $S_1$'s subtree cut condition (\ref{eq:S}).
\end{proof}


\subsubsection{Reducing to the Two Subtrees Case}

Let $v_0$ be the leaf of the second branch guaranteed by \Cref{lem:second-branch}. To find all extreme sets in the form of complement of two subtrees, $\overline{u^\downarrow\cup v^\downarrow}$, we only need to find the extreme sets with two endpoints $u$ and $v$ on the branches of $u_0$ and $v_0$.
Now we can contract edges except those on $u_0^\uparrow$ and $v_0^\uparrow$, so that the tree only consists of two branches. Split the two branches by deleting the tree edge incident to the root on the branch of $v_0$. Then, contract $u_0$ and $v_0$ into a single vertex, and declare it as the new root; see \Cref{fig:2-branches}. It is easy to see that any extreme set that was previously of the form $\overline{u^\uparrow\cup v^\uparrow}$ for $u\in u_0^\uparrow$ and $v\in v_0^\uparrow$ is now a union of two subtrees, or just one subtree if $v$ is a child of root. Therefore, we have reduced to the two subtrees case, as desired.


\begin{figure}[t]
\centering
\includegraphics[width=.45\textwidth]{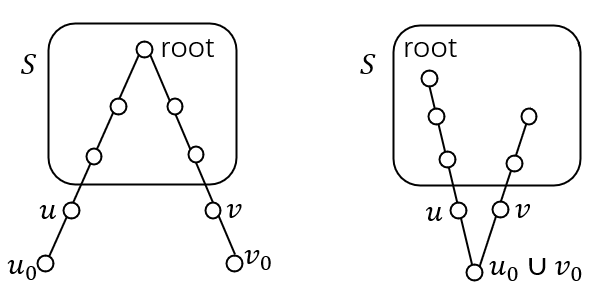}
\caption{Reduction from complement case to subtrees case after fixing 2 branches}
\label{fig:2-branches}
\end{figure}

\section{Merging Two Laminar Trees}
\label{sec:merge}


In this section, we prove the lemma that merges two laminar families and preserves all extreme sets in both families.

\begin{lemma}
\label{lem:merge}
Given two laminar families $\mathcal X$ and $\mathcal Y$ on the vertex sets, \Cref{alg:verify} constructs a merged laminar family $R$ containing all extreme sets in $\mathcal X\cup \mathcal Y$ (and possibly other sets in $\mathcal X\cup \mathcal Y$) in $\tilde{O}(m)$ time.
\end{lemma}

We represent each laminar family by a tree to ensure its representation size is linear and not quadratic. In the tree representation, each node corresponds to a set in the family, except the root which represents all vertices $V$. Each node $x$ has a (possibly empty) set of vertices in $V$ associated with it, and the corresponding set in the laminar family is all vertices in $V$ associated with any node in the subtree rooted at $x$. Each vertex in $V$ is associated with exactly one node. Note that we do not require that only leaves have a nonempty set of associated vertices. This is because even if we start with a tree with only nonempty sets at leaves, the algorithm's operations on the tree may produce internal vertices with nonempty sets. 

Our algorithm requires the definition of a \emph{bough} of a tree, as follows.

\begin{defn}
A bough is a tree path that starts at a leaf, extends toward the root and stops before reaching the first node with more than one children.
\end{defn}

Our algorithm decomposes the laminar trees into disjoint boughs. Initially all vertices are in the leaves. But as we proceed, the boughs will be contracted to their parents, so there may be vertices in internal nodes.

\subsection{Removing Inconsistent Sets}

We start by analyzing \Cref{alg:verify}.
\begin{defn}
We define a set $U\subseteq V$ to be \emph{consistent} with $W\subseteq V$ if $\delta(U\setminus W)>\delta(U)$ or $W$ is disjoint from $U$. We define $U$ to be \emph{consistent} with a laminar family $\mathcal Y$ if $U$ is consistent with all $W\in \mathcal Y$.
\end{defn}
\begin{fact}
\label{fact:consist}
An extreme set is consistent with any vertex set.
\end{fact}

\begin{algorithm}
\caption{Verify($\mathcal X,\mathcal Y$)}
\label{alg:verify}
\SetKwInOut{Input}{Input}
\SetKwInOut{Output}{Output}
\Input{Laminar trees $\mathcal X$ and $\mathcal Y$ on vertex set $V$.}
\Output{Laminar tree of $\mathcal X^*=\mathcal X\setminus\{U\in \mathcal X:\exists W\in \mathcal Y,\, \delta(U\setminus W)\le\delta(U)\text{ and } U\cap W\ne \emptyset\}$.}
Let $\mathcal X^*=\mathcal X$.\\
\While{$\mathcal X$ is nonempty}{
 \ForEach{bough $\mathcal B$ of $\mathcal X$}{
  Using \Cref{lem:verify}, find all sets $U\in \mathcal B$ such that $\exists W\in \mathcal Y,\, \delta(U\setminus W)\le\delta(U)$ and $U\cap W\ne \emptyset$. Add all sets in $\mathcal B$ to $\mathcal X^*$ except for the ones we found.
  
  Remove the bough from $\mathcal X$, and contract the vertices in the bough to the bough's parent node.
 }
}
\Return{$\mathcal X^*$}.
\end{algorithm}

\begin{lemma}\label{lem:verify}
Consider a bough of $\mathcal X$ consisting of nested sets $U_1\subseteq U_2\subseteq\cdots\subseteq U_k$. There is an algorithm that outputs all sets $U_i$ for which there exists $W\in\mathcal Y$ with $\delta(U\setminus W)\le\delta(U)$ and $U\cap W\ne\emptyset$. The algorithm takes $\tilde{O}(m)$ preprocessing time and then handles each bough in time proportional to (up to polylogarithmic factor) the size of the induced subgraph $G[U_k]$.
\end{lemma}
\begin{proof}
Initialize a dynamic tree on the tree $T$ representing laminar family $\mathcal Y$ with initial value $0$ on each node, along with a Boolean flag that is initially false. Our goal is to maintain, for each $W\in\mathcal Y$, the value $\delta(U\setminus W)-\delta(U)$. We are interested in whether this value is at most $0$ for all $W$ with $U\cap W\ne\emptyset$. Throughout, we abuse notation by referring to each node and its set $W$ interchangeably.

Iterate through $U_1,U_2,\ldots,U_k$ in that order. For each $U_i$, we loop through the vertices $u\in U_i\setminus U_{i-1}$ one by one in arbitrary order $u_1,u_2,\ldots,u_\ell$. For convenience, define $U_{i-1,j}=U_{i-1}\cup\{u_1,\ldots,u_j\}$. For each vertex $u_j$, do the following.
 \begin{enumerate}
 \item For each incident edge $(u_j,v)$ where $v\in U_{i-1,j-1}$, add twice the weight to each node on the path from $v$ to the lowest common ancestor of $u_j$ and $v$ (excluding the LCA).
 \item Add $2\delta(u_j, U_{i-1,j-1})-\delta(u_j)$ to each node on the path from $u_j$ to root. Each such node has its Boolean flag set to true.
 \end{enumerate}
After these operations, the algorithm queries the minimum value over all nodes whose flag is set to true. If this minimum value is at most $0$, then we add $U_i$ to the output set.

For convenience, define $U_0=\emptyset$. We prove by induction on $i\ge 0$ that after processing $U_i$, each node in the dynamic tree carries value $\delta(U_i\setminus W)-\delta(U_i)$. This is vacuously true for $i=0$ since $U_0=\emptyset$ and each node carries value $0$. To prove each inductive step, we perform a separate induction on vertices $u_1,\ldots,u_{j-1}\in U_i\setminus U_{i-1}$.
We claim that after inserting $u_j$, each node in the dynamic tree with corresponding set $W\subseteq V$ carries value $\delta(U_{i-1,j}\setminus W)-\delta(U_{i-1,j})$ where $U_{i-1,j}=U_{i-1}\cup\{u_1,\ldots,u_j\}=U_{i-1,j-1}\cup\{u_j\}$. This is true for $j=0$ by induction on $i-1$.
For each set $W$, consider the change of its value after adding $u_j$ into $U_{i-1,j-1}$, that is \[\Delta_j=(\delta(U_{i-1,j}\setminus W)-\delta(U_{i-1,j}))-(\delta(U_{i-1,j-1}\setminus W)-\delta(U_{i-1,j-1})).\]
There are two cases. If $u_j\notin W$,
\[\Delta_j=2\delta(u_j, U_{i-1,j-1}\cap W).\]
If $u_j\in W$,
\[\Delta_j =2\delta(u_j, U_{i-1,j-1})-\delta(u_j).\]

We show that this change is correctly accounted for in the dynamic tree updates. For each set $W$ not containing $u_j$, each edge $(u_j, v)$ with $v\in U_{i-1,j-1}\cap W$ has its weight added twice to the value of $W$, since $W$ as an ancestor of $v$ but not an ancestor of $u_j$ on $T$. Therefore the value of $W$ is increased by $2\delta(u_j, U_{i-1,j-1}\cap W)$, as expected. Note that in step (1), for each edge $(u_j, v)$, we only add its weight to sets not containing $u_j$.
For each set $W$ containing $u_j$, it lies on the path from $u_j$ to the root, and its value is increased by $2\delta(u_j, U_{i-1,j-1})-\delta(u_j)$ in step (2), as expected.
These changes match the required net change $\Delta_j$.


It remains to show that a set $U_i$ should be output if and only if there is a node in the tree with value at most $0$ and Boolean flag set to true. We have already shown that any node $W$ of value at most $0$ satisfies $\delta(U\setminus W)-\delta(U)\le 0$, so it remains to show that a node $W$ is flagged true if and only if $U\cap W\ne\emptyset$. Observe that for each vertex $u_j$ processed, we flag the nodes from $u_j$ to the root as true; their sets are precisely those that contain $u_j$. Since the sets $U_i$ are nested, once we finished processing $U_i$, the nodes $u_j$ we have processed on iterations up to $i$ are precisely $U_i$. In other words, a set $W$ is flagged true if and only if $U\cap W\ne\emptyset$, as desired.

Finally, we discuss running time. All dynamic tree operations take $O(\log n)$ time. The total number of edges $(u_j,v)$ for $v\in U_{i-1,j-1}$, summed over all $i$ and $j$, is at most the number of edges in the induced subgraph $G[U_k]$.
\end{proof}
\begin{lemma}
\Cref{alg:verify} takes $\tilde O(m)$ time.
\end{lemma}
\begin{proof}
For each bough with root $U_k$, we spend time proportional to the number of edges in induced graph $G[U_k]$, and then we contract all vertices in $U_k$ into a single vertex. The contraction removes all edges in the induced graph $G[U_k]$, so the decrease in number of edges pays for the processing time of the bough. Since there are $m$ initial edges, the total running time becomes $\tilde O(m)$.
\end{proof}

\begin{corollary}
\label{cor:merge-laminar}
Given two laminar families $\mathcal X$ and $\mathcal Y$, let $\mathcal X^*=\text{Verify}(\mathcal X,\mathcal Y)$ and $\mathcal Y^*=\text{Verify}(\mathcal Y,\mathcal X)$. Then, $\mathcal X^*\cup \mathcal Y^*$ is laminar.
\end{corollary}
\begin{proof}
Assume for contradiction that some $U\in \mathcal X^*\subseteq \mathcal X$ crosses some $W\in \mathcal Y^*\subseteq \mathcal Y$. Then $U\cap W\ne \emptyset$, and by posi-modularity, either $\delta(U\setminus W)\le \delta(U)$ or $\delta(W\setminus U)\le \delta(W)$. By \Cref{lem:verify}, either $U$ or $W$ will be removed in \Cref{alg:verify}, which contradicts the definitions of $\mathcal X^*$ and $\mathcal Y^*$.
\end{proof}

Given laminar families $\mathcal X,\mathcal Y$ and their tree structures, we can therefore run \Cref{alg:verify} to obtain $\mathcal X^*,\mathcal Y^*$ such that $\mathcal X^*\cup\mathcal Y^*$ is a laminar family containing all extreme sets in $\mathcal X\cup\mathcal Y$. We can easily recover the tree structures of $\mathcal X^*$ and $\mathcal Y^*$ as well. It remains to recover the tree structure of $\mathcal X^*\cup\mathcal Y^*$.
\begin{lemma}
Assume that $\mathcal X^*$, $\mathcal Y^*$, and $\mathcal X^*\cup\mathcal Y^*$ are all laminar families. There is an $O(n \log n)$ algorithm that computes the tree structure of $\mathcal X^*\cup\mathcal Y^*$.
\end{lemma}
\begin{proof}
Let $T_{\mathcal X^*}$ and $T_{\mathcal Y^*}$ be the tree structures for $\mathcal X^*$ and $\mathcal Y^*$, respectively.
We first find, for each set $Z\in\mathcal X^*\cup\mathcal Y^*$, (a pointer to) the parent node of $Z$ in the tree structure $T$ of $\mathcal X^*\cup\mathcal Y^*$. Pick an arbitrary vertex $z\in Z$. Since $\mathcal X^*\cup\mathcal Y^*$ is laminar, the parent of $Z$ is exactly the set $Z'\in\mathcal X^*\cup\mathcal Y^*$ satisfying $z\in Z'$ and $|Z'|>|Z|$ and $|Z'|$ is as small as possible given these two constraints. The set $Z'$ can be found by computing a binary search on the path from $z$ to the root on the tree structures for $\mathcal X^*$ and $\mathcal Y^*$ and taking the best $Z'$ found. If there is a tie, as in both $\mathcal X^*$ and $\mathcal Y^*$ include the parent $Z'$, then we take the pointer to the one in $\mathcal X^*$, and we can ignore the duplicate one in $\mathcal Y^*$ in the next step of the algorithm.

By computing all the parents (and ignoring the duplicate nodes), we can build the tree $T$ for $\mathcal X^*\cup\mathcal Y^*$ where each node corresponds to the same set as its pointer in the tree structure of $\mathcal X^*$ or $\mathcal Y^*$. It remains to compute the set of vertices associated with each node. For each node $Z\in\mathcal X^*$ or $Z\in\mathcal Y^*$ with children $Z_1,\ldots,Z_k$ in $T_{\mathcal X^*}$ or $T_{\mathcal Y^*}$ respectively, we check whether each vertex in $Z\setminus\bigcup_iZ_i$ is in any child of $Z$ in $T$. This can be done by first marking the pointer of each child of $Z$ in $T$ (which is a node in $T_{\mathcal X^*}$ or $T_{\mathcal Y^*}$), and then testing, for each vertex $v\in Z\setminus\bigcup_iZ_i$ and for both $T_{\mathcal X^*}$ and $T_{\mathcal Y^*}$, whether the node associated with $v$ in either $T_{\mathcal X^*}$ or $T_{\mathcal Y^*}$ is a descendant of a marked node. This can be done in $O(\log n)$ time per vertex $v$ using tree data structures. Any vertex $v\in Z$ that is not a descendant of any marked node is associated with $Z$ in the new tree $T$. As for running time, marking the pointer of each child of $Z$ in $T$ takes $O(\log n)$ times the number of children, which is $O(n \log n)$ time summed over all $Z$. Also, we can iterate through $v\in Z\setminus\bigcup_iZ_i$ since these are precisely the nodes associated with $Z$ in either $T_{\mathcal X^*}$ or $T_{\mathcal Y^*}$, and the descendant queries take $O(|Z\setminus\bigcup_iZ_i|\log n)$ time overall, which again sums to $O(n \log n)$ over all $Z$. 
\end{proof}

\clearpage

\bibliographystyle{alpha}
\bibliography{ref}

\end{document}